\documentclass[pra,aps,twocolumn,superscriptaddress,dvipsnames,amssymb,9pt,floatfix]{revtex4}
\usepackage{graphicx}
\usepackage{amsmath,amssymb,amsthm,amsfonts,mathptmx}
\usepackage{color}
\usepackage[matrix,frame,arrow]{xypic}
\usepackage{times}
\usepackage{algorithm}

\usepackage{algcompatible}

\usepackage{longtable}
\usepackage{tikz}

\theoremstyle{plain}

\newtheorem{theorem}{Claim}

\theoremstyle{definition}

\newcommand{\ket}[1]{|#1\rangle}
\newcommand{\braket}[2]{\langle#1|#2\rangle}

\DeclareMathOperator{\CZ}{CZ}
\DeclareMathOperator{\FX}{XNOT}
\DeclareMathOperator{\CNOT}{CNOT}

\DeclareMathOperator{\IAC}{IAC}
\DeclareMathOperator{\IEC}{IEC}

\begin{document}

\title{Fast graph operations in quantum computation}
\author{Liming Zhao}
\affiliation{Singapore University of Technology and Design, 20 Dover Drive, Singapore 138682}
\author{Carlos A. P\'erez-Delgado}
\affiliation{Singapore University of Technology and Design, 20 Dover Drive, Singapore 138682}
\author{Joseph F. Fitzsimons}\email{joseph_fitzsimons@sutd.edu.sg}
\affiliation{Singapore University of Technology and Design, 20 Dover Drive, Singapore 138682}
\affiliation{Centre for Quantum Technologies, National University of Singapore, 3 Science Drive 2, Singapore 117543}
\begin{abstract}
The connection between certain entangled states and graphs has been heavily studied in the context of measurement-based quantum computation as a tool for understanding entanglement. Here we show that this correspondence can be harnessed in the reverse direction to yield a graph data structure which allows for more efficient manipulation and comparison of graphs than any possible classical structure. We introduce efficient algorithms for many transformation and comparison operations on graphs represented as graph states, and prove that no classical data structure can have similar performance for the full set of operations studied.
\end{abstract}
\maketitle

\section{Introduction}

A key motivation for the development of quantum information processing technologies is the promise of faster algorithms for computational problems. Such algorithms have been exemplified by the discovery of a polynomial time quantum algorithm for factoring integers \cite{shor1}, as well as algorithms for search \cite{grover1997} and other problems which offer a provable advantage over their classical counterparts. Despite these successes, progress in developing quantum algorithms remained relatively slow for a number of years. A common feature of many early quantum algorithms is that although they all act in quantum ways procedurally, the data acted upon is essentially classical. More recently, several algorithms have emerged which achieve greater efficiency than their classical counterparts by representing problem instances in a fundamentally quantum manner within the computation, as is the case in the quantum algorithms for linear algebra \cite{harrow2009quantum} and estimating knot invariants \cite{aharonov2009polynomial}. As the representation of the problem in terms of quantum states appears fundamental to the success of these algorithms, the question then arises as to what uniquely quantum data structures exist.

In this paper we explore data structures for graphs based on quantum states which confer a provable advantage over classical data structures in terms of the efficiency with which the graphs can be manipulated and compared. We draw on the well studied connection between graphs and certain entangled states, known as graph states, which have been heavily studied in the context of measurement-based quantum computation \cite{raussendorf2001one,raussendorf2003measurement,briegel2009measurement,browne2006oneway}. However, whereas most prior work has explored  the relationship from the point of view of utilising graphs to represent quantum states \cite{hein2004multiparty,van2004graphical,van2004efficient,hein2006entanglement}, we proceed in the reverse direction and consider the use of graph states as a representation of graphs, an approach which has so far received little attention.

\section{Notation and preliminary definitions}

In this paper we use the standard notation for graphs commonly used in both physics and computer science. We understand a graph to be a tuple $G = \{V,E\}$, where the first member $V$ is a set of vertices and the second member is a subset of the cartesian product of the first with itself, $E \subseteq V \times V$, corresponding to the set of edges. We will focus only on \emph{undirected} graphs, where if $\{v_i, v_j\} \in E$ then also $\{v_j, v_i\} \in E$.
For any vertex $v_0 \in V$, we will call $v_i$ a \emph{neighbour} of $v_0$ if $\{v_0,v_i\} \in E$. The set of all neighbours of a vertex $v_0$ is referred to as its neighbourhood and is denoted by $N(v_0)$. The number of edges incident on a vertex $v_0$ is called its \emph{degree}.
Sometimes, in graph theory, graphs with multiple edges between the same pair of vertices, or with self-loops (edges from one vertex to itself) are considered. These are often referred to as \emph{multigraphs}, whereas graphs without self-loops ore multiple edges are called \emph{simple graphs}. Here, we will consider multigraphs with (singular) self-loops, but will not consider multiple edges.

 Given a graph $G$, a quantum graph state $\ket{G}$ representing $G$ is a pure quantum state on $N = |V|$ qubits, such that 
\begin{equation}\label{eq:one}
 K_{G}^{v}\vert G\rangle=\vert G\rangle\quad \forall v \in V, 
 \end{equation} 
  where $K_{G}^{v}=X_{v}\prod_{u\in N(v)} Z_u$, where $X_u$ ($Z_u$) denotes the Pauli $X$ (Pauli $Z$) operator acting on the qubit labelled $u$ \cite{hein2004multiparty}. The set of operators $K_{G}^{v}$ generate a group, called the \emph{stabiliser}, of operators that act trivially on $\ket{G}$.
  
Intuitively, the most common way to think of a graph state is to picture an $N$-qubit system, one qubit representing each vertex, initially in the state $\ket{+}^{\otimes N}$ (where $\ket{+} = \frac{1}{\sqrt{2}}\ket{0} + \frac{1}{\sqrt{2}}\ket{1}$), to which are then applied controlled-$Z$ ($\CZ$) operators between every two qubits representing neighbouring vertices in the graph $G$. In fact, one can prepare graph states in such a way. This is known as the constructive definition of a graph state \cite{browne2006oneway}, and is equivalent to the stabiliser definition. From this latter definition it is tempting to believe that different graphs necessarily lead to states which are entangled in different ways. This is not the case, however, as graph states generated from different graphs may be equivalent up to local operations. If a graph state $\ket{\tilde{G}}$ can be obtained from a graph state $\ket{G}$ by applying solely local Clifford group operators, then the two graph states are said to be LC-equivalent \cite{van2005local}.

In this paper we will extend the usual definitions of graph states to include graphs with edges which connect a vertex to itself. In the constructive definition, such a self-edge will take the form of a $Z$ gate applied to the qubit corresponding to the vertex in question. Since this is effectively a CZ gate with the same control and target qubit, we will use the convention $\text{CZ}_{i,i}$ is taken to mean $Z_i$. The generators of the stabilizer for a graph containing self-edges are then given by
\begin{equation}\label{eq:stabgen}
K_{G}^{v}=(-1)^{|N(v) \cap \{v\}|}X_{v}\prod_{u\in N(v)\setminus u} Z_u.
 \end{equation}
 
In this paper we will be using heavily an operation not usually studied in classical graph data-structures, that we refer to as \emph{edge-complementation.} This operation adds an edge $(u,v)$ if it  does not currently exist in the graph, and removes if it does. We will find it useful to introduce mathematical notation for this operation. If $G = \{V,E\}$, then $G' = \{V,E \oplus (u,v)\}$ will be understood to be the resulting graph after complementing the edge $(u,v)$ in $G$.

\section{Representation of data}

Given the one to one correspondence between graph states and graphs, it is natural to think of such states as a quantum representation of the corresponding graph. A barrier to directly using such states in place of classical data structures, such as adjacency matrices or edge lists \cite{mchugh1990algorithmic}, is that graph states are not orthogonal to one another, and hence it is not possible to manipulate them in an arbitrary manner. In order to operate non-destructively on non-orthogonal quantum states, it is necessary that the applied operations act unitarily on the space spanned by the graph states. Thus, only reversible manipulations of the underlying graph are allowed. Furthermore, since the operations cannot change the inner product between graph states, not all reversible manipulations of the graph are allowed. Despite these constraints, however, the structure of graph states allows for a wide range of manipulations of the underlying graph via unitary gates. 

In fact, since a graph state requires only $N$ qubits to represent any $N$ vertex graph, it is more efficient than any classical data structure, while still allowing for comparison operations both between graphs and between subsets of vertices within the same graph. This is because any classical data structure which can uniquely identify an arbitrary $N$-vertex undirected graph requires at least one bit per edge in the graph, and so requires a minimum of $N^2$ bits.

In what follows we present the basic manipulations and comparison operations possible on graphs represented as graph states. In counting the resources required for each operation we consider the number of elementary gates required in the circuit model, for an arbitrary universal gate set of one- and two-qubit gates with no locality restrictions.

\subsection{Preparation of graph states}

Given a classical description of a graph $G=(V,E)$, the constructive definition of graph states provides a prescription for how to prepare the corresponding quantum state: prepare $|V|$ qubits in state $\ket{+}$, each corresponding to a vertex in the graph, and then apply $\CZ$ operations between any pair corresponding to an element of $E$. This takes a total of $O(|V|+|E|)$. Many graph states can, in fact, be prepared more efficiently by making use of some combination of local complementation and edge operations as described later. Perhaps the best example of this is the graph state corresponding to the complete graph. Following the above procedure required $O(|V|^2)$ operations, whereas the same state can be prepared in $O(|V|)$ steps by applying the intraset complementation procedure described in Section \ref{sec:intra-comp}. This should not come as a surprise, since graph state in question is locally equivalent to a GHZ state \cite{hein2004multiparty}.

The efficient preparation of graph states corresponding to predetermined graphs, using either deterministic \cite{mhalla2004complexity,hoyer2006resources,cabello2011optimal} or probabilistic operations \cite{barrett2005efficient,benjamin2005comment,duan2005efficient,chen2006efficient,benjamin2006brokered,campbell2007adaptive,campbell2007efficient,browne2008phase,matsuzaki2010probabilistic,fujii2010fault,li2010fault}, has been heavily studied in the literature. However, any process which determines the sequence of quantum operations necessary to construct the graph state after first reading in the classical description of the graph is fundamentally limited in efficiency once both classical and quantum operations are taken into account. This is because there are $\binom{|V|^2}{|E|}$ graphs with the same number of vertices and edges as $G$, in the worst case it is necessary to read at least $\log_2 \binom{|V|^2}{|E|}$ bits, and so the efficiency of the process cannot be faster than quadratic in the number of vertices for an arbitrary graph.

It is, however, possible to circumvent this limit to prepare arbitrary graph states with $O(|V|)$ operations provided that the initial classical representation of the graph is chosen appropriately. To see this we consider the problem of constructing a graph using queries to a oracle $\mathcal{O}$. We define $\mathcal{O}$ as a unitary operator which acts on a system of $|V|+1$ qubits, such that given a subset $S$ of $V$, specified by the location of ones in the first $|V|$ qubits, the oracle applies $X^{|E_S|}$ to the final qubit, where $E_S$ is the edge set of the subgraph of $G$ induced by $S$. It is possible to take advantage of the structure of $\mathcal{O}$ to create $\ket{G}$ using a single query with the state $\ket{+}^{\otimes |V|}\otimes \ket{-}$. To see this note that for a computational basis state $\ket{S}$, $\braket{S}{G} = -\sqrt{2^{|V|}}$ if the number of $\CZ$ operations applied between the qubits corresponding to the location of ones in $\ket{S}$ is odd, and otherwise $\braket{S}{G} = \sqrt{2^{|V|}}$. Thus $\ket{S}$ can be thought of as specifying a subset of the vertices of $V$ through the location of ones. Querying $\mathcal{O}$ with $\ket{+}^{\otimes |V|}\otimes \ket{-}$ results in 
\[
\mathcal{O} \ket{+}^{\otimes |V|}\otimes \ket{-} = 2^{\frac{|V|}{2}} \sum_s (-1)^{|E_S|} \ket{S},
\]
which is equal to $\ket{G}$.

\subsection{Basic manipulation}

It has previously been established that all stabilizer states are locally equivalent to graph states \cite{hein2006entanglement}, and hence Clifford operations give rise to transitions between graph states up to local operations. In this section we explore a set of basic operations which map graph states onto other graph states. In all cases these operations can be identified with a simple transformation of the graph, which can be viewed in terms of edge complementation operations. In the next section these operations will be used to construct compound operations for achieving natural operations on the graph.

\subsubsection{Tensor product}
Perhaps the simplest operation which can be performed is to consider the joint state of two graph states represented by distinct sets of qubits. From the definition of graph states, it immediately follows that the joint state of such a system is itself a graph state, with the associated graph corresponding to the union of the two graphs underlying the initial graph states. Here the set of vertices in each graph is taken to be distinct. Thus we have the rule
\begin{align}
\otimes: ~~ V\mapsto V_1 \cup V_2, ~E \mapsto E_1 \cup E_2,
\end{align}
where the two initial graphs are taken to be $G_1 = (V_1,E_1)$ and $G_2 = (V_2,E_2)$.

Since this operation applies for any pair of graphs, vertices can be added to a graph simply by preparing ancilla qubits in the state $\ket{+}$, which corresponds to graph state for a single disconnected vertex. Barring considerations like limited memory, garbage collection, or quantum specific issues like cooling and fault tolerance this operation can be reasonably be expected to take $O(1)$ in most quantum computing architectures.

On the other hand, removal of a vertex cannot be accomplished deterministically, since this is a non-reversible operation. However, it is possible to delete a vertex with constant probability of success. In order to delete a vertex we proceed by measuring it in the $Z$ basis. If the measurement output is $0$, then the operation succeeds. Otherwise, in order to bring the state of the data structure back into a graph state one would need to apply a $Z$ correction on all neighbouring vertices. This is simply an application of the Pauli measurement rules studied in \cite{hein2004multiparty}. Without access to a classical structure (or some other method) detailing which vertices neighbour which, the $Z$ correction becomes impossible. Thus, we consider vertex deletion to be a probabilistic method. If one expects the need to delete vertices, then multiple copies of the graph state would usually be needed. However, as we shall see later, there are certain circumstances in which this operation is useful as part of more complex deterministic operations.

\subsubsection{Pauli operations}
In our extension of graph states to include self-loops, such loops were represented in the constructive definition by an additional $Z$ gate applied to the qubit corresponding to the vertex in question. Thus, since Pauli operators are self-adjoint, applying $Z_a$ to the graph has the effect of adding a self-loop if one is not present on vertex $a$, and removing one if it was initially present. We will refer to this graph operation as \emph{loop-complementation}. Furthermore, since the stabilizer contains an element $K_G^a$, which is proportional to the product of $X_a$ with $Z$ operators on the neighbourhood of $a$ (less any self-loops), applying $X_a$ is equivalent, up to global phase, to applying $Z_v$ for every $v\in N(a)\setminus a$. The effect of this on the underlying graph is to perform loop-complementation on the neighbourhood of $a$, excluding $a$ itself. Since $Y_a$ is proportional to the product of $Z_a$ and $X_a$, the effect of $Y_a$ on the underlying graph is simply to apply loop-complementation to every vertex in $N(a)\cup a$. These operations can be expressed via their action on the vertex and edge sets of the underlying graph as follows:
\begin{align}
X_a:&~~ V\mapsto V, ~E\mapsto E \bigoplus_{v\in N(a)\setminus a} (v,v)\\
Y_a:&~~ V\mapsto V, ~E\mapsto E \bigoplus_{v\in N(a)\cup a} (v,v)\\
Z_a:&~~ V\mapsto V, ~E\mapsto E\oplus (a,a)
\end{align}

\subsubsection{$\CZ$ operations}
The next operation we consider is a $\CZ$ gate applied between qubits $a$ and $b$ in the graph state. When considered from the point of view of the constructive definition, the application of a $\CZ$ operation to an existing graph state simply cancels the $\CZ$ from a pre-existing edge if present, removing it. Otherwise, it can simply be considered an extension of the state preparation corresponding to the addition of an edge between the chosen vertices. This is because the $\CZ$ operation is self-adjoint, and hence performing two such operations between a pair of qubits is equivalent to performing the identity operation. Thus the $\CZ$ operation implements the transformation
\begin{align}
\CZ_{a,b}: & ~~V \mapsto V, ~E \mapsto E \oplus (a,b).
\end{align}

This $\CZ$ operation then corresponds to the operation of \emph{edge complementation}, which interchanges the presence and absence of a particular edge. Aside from complementing the edge between a pair of vertices, it is also possible to efficiently complement larger sets of edges efficiently, as we shall see next.

\subsubsection{$\CNOT$ operations}
Given two vertices $a$ and $b$, the effect of a $\CNOT$ operator controlled by $a$ and targeted on $b$ is to complement the edges between $a$ and the neighbourhood of $b$ (other than $b$ itself). In order to prove that this is indeed the case, we consider the stabiliser for an arbitrary graph $G$. If the corresponding graph state $\ket{G}$ undergoes a unitary transformation $U$, the generators of the stabiliser group for the new state $U\ket{G}$ are given by $UK_G^v U^\dagger$ for $1\leq v\leq N$ \cite{browne2006oneway}. Since the unitary transformation applied in this case is a CNOT controlled by the qubit corresponding to vertex $a$ and targeted on the qubit corresponding to vertex $b$, the only generators of the stabiliser which are altered by this operation are those which either act as $Z$ on vertex $b$ (those for which $v\in N(b)$) or as $X$ on vertex $a$ (the single case of $v=a$). Thus, the transformed generators are given by
\begin{eqnarray*}
\tilde{K}_a &=&  (-1)^{|N(a)\cap \{a\}|}X_{a} \otimes X_b \left(\prod_{u\in N(a)\setminus b} Z_u\right)\left( \prod_{v\in N(a)\cap b} Z_a\otimes Z_v\right) \\
\tilde{K}_b &=& (-1)^{|N(b)\cap \{b\}|}X_b \prod_{u\in N(b)} Z_u\\
\tilde{K}_c &=& (-1)^{|N(c)\cap \{c\}|}X_c Z_a \prod_{u\in N(c)} Z_u  ~~~ \text{for }c \in N(b) \setminus (a\cup b) \\
\tilde{K}_d &=& (-1)^{|N(d)\cap \{d\}|}X_d \prod_{u\in N(d)} Z_u ~~~ \text{for }d \notin N(b)\cup a\cup b.
\end{eqnarray*}

Now, consider the graph $G'$ obtained from $G$ by complementing the edges between $a$ and the neighbours of $b$, other than $b$, and adding a self-loop to $a$ if $b$ has a self-loop. In this case $\tilde{K}_b = K^b_{G'}$, $\tilde{K}_c = K^c_{G'}$ and $\tilde{K}_d = K^d_{G'}$. Furthermore, 
\begin{multline}
\tilde{K}_a \times \tilde{K}_b = (-1)^{\gamma}X_{a} \left(\prod_{u\in (N(a)\Delta N(b))\setminus \{a,b\}} Z_u\right) = K^a_{G'},
\end{multline}
where $\gamma = |N(a)\cap \{a\}| + |N(b)\cap \{b\}|$. Hence, $\CNOT_{a,b} \ket{G} = \ket{G'}$ as stated. It is worth noting that the presence of a $(-1)$ in the stabiliser corresponding to a vertex $v$ denotes the presence of a self-loop on $v$. Thus, a CNOT operator acting vertices $a$ and $b$ affects the vertex and edge sets of the underlying graph according to the following rule:
\begin{equation}
\CNOT_{a,b} : ~~V \mapsto V,
 ~E \mapsto
  E \bigoplus_{v \in N(b) \setminus b } (a,v) \bigoplus_{v\in N(b)\cap b} (a,a) 
\end{equation}

\subsubsection{$\FX$ operations}
The final basic manipulation  we look at is the $\FX$ quantum operator, defined as $\FX_{a,b} = (H_a\otimes H_b) \CZ_{a,b} (H_a\otimes H_b)$.
Unlike previous basic manipulation operations, $\FX_{a,b} (G)$ is only well defined when the vertices $a$ and $b$ are \emph{not} neighbours in $G$. If $a$ and $b$ \emph{are} neighbours in $G$ then $\FX_{a,b}$ maps $G$ outside the state of valid graph states, and is therefore left undefined as an operation on graphs. However, if $(a,b)$ is \emph{not} an edge in $G$, the operation has the effect of  complementing all edges between the neighbourhoods of the two vertices, $a$ and $b$. 

More formally, take $C = (N(a)\setminus (N(b) \cup \{a\}))$, $D = (N(b)\setminus (N(a)\cup\{b\}))$ and $F = (N(a)\cap N(b)\setminus (\{a\}\cup\{b\}))$, and assume that there is no edge between vertices $a$ and $b$. 
Then, for every $v \in C \cup F$ and every $u \in D \cup F$, the operator $\FX$ complements the edge $(v,u)$. This, implies that all $v \in F$ will be self-loop complemented. The operator $\FX$ also has the effect that if $a$ ($b$) has a self-loop then  all elements of $D \cup F$ ($C \cup F$) will be self-loop complemented. Also, because every edge $(v,u)$ for $u,v \in F$ is acted on \emph{twice}, the net effect is that no edge is complemented in this case.
In short, the operator $\FX$ has the following effect regarding self-loops. If neither $a$ nor $b$ have self-loops, then $\FX$ complements self-loops on $F$. If $a$ ($b$) has a self-loop but not $b$ ($a$), then self-loops are complemented on $D$ ($C$). Finally, if both $a$ and $b$ have self-loops then self-loops are complemented on all of $C$, $D$ and $F$.

The above can proven as follows. Since there is no edge between $a$ and $b$, $K_G^a$ and $K_G^b$ commute with $\FX$ and hence are also generators of the stabiliser of $FZ_{ab}\ket{G}$. Therefore, $\tilde{K}_a = K_G^a$ and $\tilde{K}_b = K_G^b$. The same is true for $K_G^u$ for all $u \notin N(a)\cup N(b)$. Thus the only generators of the stabiliser altered by $\FX_{a,b}$ are those corresponding to vertices in $C$, $D$ or $E$.

For a vertex $c \in C$, $K_G^c$ is transformed to 
\begin{eqnarray*}
\tilde{K}_c &=& (-1)^{|N(c)\cap \{c\}|} X_c X_b \prod_{u \in N(c)} Z_u\\
&=& (-1)^{|N(c)\cap \{c\}| + |N(b)\cap \{b\}|} \tilde{K}_b X_c \prod_{u \in N(c)} Z_u \prod_{v \in N(b) \setminus b} Z_v.
\end{eqnarray*}
Similarly, for $d\in D$
\begin{eqnarray*}
\tilde{K}_d &=& (-1)^{|N(d)\cap \{d\}|} X_d X_a \prod_{u \in N(d)} Z_u\\
&=& (-1)^{|N(d)\cap \{d\}| + |N(a)\cap \{a\}|} \tilde{K}_a X_d \prod_{u \in N(d)} Z_u \prod_{v \in N(a)} Z_v,
\end{eqnarray*}
and for $f\in F$
\begin{eqnarray*}
\tilde{K}_f &=& (-1)^{|N(f)\cap \{f\}| } X_f Y_a Y_b \prod_{u \in N(f) \setminus (a \cup b)} Z_u\\
&=& (-1)^{\gamma}\tilde{K}_a \tilde{K}_b X_f \prod_{u \in N(f)} Z_u \prod_{v \in N(a)} Z_v \prod_{w \in N(b)} Z_w,
\end{eqnarray*}
where $\gamma = 1+  |N(f)\cap \{f\}| + |N(a)\cap \{a\}| + |N(b)\cap \{b\}|$. Taking $G'$ to be the graph obtained from $G$ by complementing every edge which connects a vertex in $N(a)$ to a vertex in $N(b)$, and performing the self-loop complementations described before, we have
\begin{eqnarray*}
K_{G'}^a &=& \tilde{K}_a\\
K_{G'}^b &=& \tilde{K}_b\\
K_{G'}^u &=&\tilde{K}_u ~~~ \text{for } u\notin N(a) \cup N(b) \\
K_{G'}^c &=& \tilde{K}_b \tilde{K}_c ~~~ \text{for } c\in C \\
K_{G'}^d &=& \tilde{K}_a \tilde{K}_d ~~~ \text{for } d\in D \\
K_{G'}^f &=& \tilde{K}_a \tilde{K}_b \tilde{K}_f ~~~ \text{for } f\in F.
\end{eqnarray*}

Hence, $\FX_{a,b}$ has the following effect,
\begin{align}
&\FX_{a,b} :~~ V \mapsto V, \notag\\
&E \mapsto E \bigoplus_{\substack{v \in C \cup F\\u \in D \cup F}} (v,u) \bigoplus_{\substack{v \in C \cup F\\ u\in N(b)\cap \{b\}}} (v,v) \bigoplus_{\substack{v \in D \cup F\\  u\in N(a)\cap \{a\}}} (v,v) 
\end{align}

However, it is important to note that this operation only results in a valid graph state provided that there is no edge between vertices $a$ and $b$. If such an edge exists, it must be removed before the $\FX$ is applied and added again afterwards, so that the total operation applied is $\CZ_{a,b} \FX_{a,b} \CZ_{a,b}$. This means that when the presence of an edge is known, such as may be the case when dealing with bipartite graphs, it is possible to complement the neighbourhoods of two vertices with a single operation.

\subsection{Ancilla driven manipulations}

It is possible to extend the set of useful operations by temporarily enlarging the graph, and then carefully choosing operations so that the new vertices are disconnected from the rest of the graph at the end of the manipulation. In such a case, the qubits corresponding to the ancillary vertices are left in a product state with the rest of the graph state, and can be discarded while leaving the remain qubits in a valid graph state.

\subsubsection{Interset complementation}
Given two disjoint sets of vertices $S_1$ and $S_2$, it is possible to implement an operator $\IAC$ that complements the edges connecting vertices in $S_1$ to vertices $S_2$ using only $O(|S_1|+|S_2|)$ operations. 
\begin{equation}
\IAC_{S_1,S_2} :~~ V \mapsto V,~~
  E \mapsto 
  E \bigoplus_{v \in S_1, u \in S_2} (v,u). 
\end{equation}

To achieve this, two ancillary vertices, $a$ and $b$, are added to the graph. Edge complementation is then performed between $a$ and each vertex in $S_1$. Since $a$ is initially disconnected, this has the result of adding an edge from $a$ to each vertex in $S_1$. Similarly, edge complementation is performed between $b$ and each vertex in $S_2$. After these operations have been applied, $N(a) = S_1$ and $N(b) = S_2$. Since $a\notin N(b)$ and $b\notin N(a)$, applying $\FX_{a,b}$ complements the edges between $S_1$ and $S_2$.   If edge complementation is then performed between $a$ and $S_1$ and between $b$ and $S_2$, all edges to vertices $a$ and $b$ are removed, resulting in the desired set-set complementation being performed. The ancilla vertices have no incident edges, and hence from the constructive definition of graph states must be in the state $\ket{+}$, unentangled with the rest of the graph. Thus the ancilla vertices can be deterministically removed from the graph simply by removing the corresponding qubits from the graph state.

\subsubsection{Intraset complementation\label{sec:intra-comp}}
Given a set of vertices $S$, it is possible to implement an operator $\IEC$ that complements edges between all distinct vertices in $S$ using only $O(|S|)$ operations:
\begin{equation}
\IEC_{S} :~~ V \mapsto V,~~
  E \mapsto 
  E \bigoplus_{v \neq u \in S} (v,u). 
\end{equation}

In order to achieve this, we make use of the local complementation operation \cite{van2004graphical,hein2004multiparty},
\begin{equation*}
U_a=\sqrt{-iX}_{a}\bigotimes_{b \in N(a)} \sqrt{iZ}_b.
\end{equation*}
When applied to a graph state $\ket{G}$, this operation results in a new graph state $\ket{G'}$ where $G'$ is the graph obtained from $G$ by complementing  edges between all pairs of distinct vertices in $N(a)\setminus a$. At the same time, if $a$ has a self-loop, then all vertices in $N(a)\setminus a$ are self-loop complemented. Due to the presence of $N(a)$ in the description of $U_a$, this operation cannot generally be used to implement local complementation on a graph state, unless $N(a)$ is known.

\begin{figure}[!t] 
\includegraphics{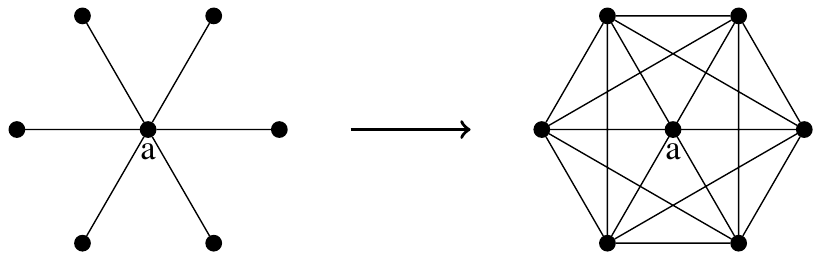}
\caption{Local complementation}
\end{figure}

Intraset complementation can be achieved by adding a new vertex $a$, and then performing edge complementation to each vertex in $S$. As in the interset complementation procedure, this has the effect of adding an edge between $a$ and each vertex in $S$, and hence $N(a)=S$. Then, $U_a$ can be applied, since $N(a)$ is known. Since this performs local complementation about the vertex $a$, this has the result of complementing all edges between all pairs of distinct vertices in $S$. Edge complementation can then be performed between $a$ and each vertex in $S$ to remove all edges incident on $a$, allowing the qubit corresponding to $a$ to be removed from the graph state, as before. The result, then, is a new graph state where all edges between distinct vertices in $S$ have been complemented.

\subsection{Comparison operations}

\subsubsection{Comparing graphs}

Given two graph states representing graphs $G_1$ and $G_2$, an important operation is to determine whether or not these two graphs are equal. As the graph states used to represent graphs are not orthogonal, this cannot be achieved deterministically. However, since every graph state is a stabiliser state, the overlap between graph states corresponding to distinct graphs, $G$ and $G'$, is bounded by $|\braket{G'}{G}|\leq \frac{1}{\sqrt{2}}$ \cite{aaronson2004improved}. As a result, it is possible to determine whether two graphs are equal or not, with constant bounded error, using $O(N)$ operations. 

The graphs represented by two graph states $\ket{G}$ and $\ket{G'}$ are necessarily different if the number of qubits comprising each state is different, since there is a one to one correspondence between qubits and vertices. As such, it suffices to provide a test only for graphs of equal size. This is achieved by means of a controlled-SWAP test \cite{NielsenChuang} as follows. An ancilla qubit is prepared in the state $\ket{+}$, and then for each vertex $i$ in $G$ a controlled-swap (Fredkin gate \cite{NielsenChuang}) gate is applied, controlled by the ancilla qubit and acting on the $i$th qubit of $\ket{G}$ and $\ket{G'}$. Together, these controlled-swap gates have the effect of swapping $\ket{G}$ and $\ket{G'}$ if the ancilla qubit is in state $\ket{1}$, and acting as the identity if the ancilla is in state $\ket{0}$. The ancillary qubit is then measured in the $X$ basis, and the graphs are deemed equal if the measurement results in $\ket{+}$.

In the case where $G=G'$, the controlled-swap procedure leaves the system invariant, since swapping the graph states is equivalent to the identity operations, and hence the ancilla is left unchanged in the state $\ket{+}$. Thus if two graph states are equal, this test will always result in a measurement of $\ket{+}$. If however $G\neq G'$, the above procedure will result in $\ket{-}$ with probability $\frac{1-|\braket{G'}{G}|^2}{2}$. Since $|\braket{G'}{G}|\leq \frac{1}{\sqrt{2}}$, this comparison operation yields an incorrect result with probability at most $\frac{3}{4}$. Thus it is possible to compare graphs via their corresponding graph states with one-sided error of at most $\frac{3}{4}$.

\subsubsection{Automorphism testing}

Given a graph state $\ket{G}$ corresponding to graph $G$, and a permutation of vertices $P$, it is possible to determine whether $P$ is an automorphism on $G$ with constant one-sided error, using a similar approach to the graph comparison test. In order to do this, a single ancilla is prepared in state $\ket{+}$, and a permutation of the qubits according to $P$, conditioned on the ancilla qubit being in state $\ket{1}$. The ancilla qubit is then measured in the $X$ basis. 

If $P$ is an automorphism on $G$, then the graph state is invariant under this permutation, and so the controlled permutation operator acts as the identity, and so the measurement must result in $+1$. If, on the other hand, $P$ is not an automorphism on $G$, the state of the system prior to measurement is $\frac{1}{\sqrt{2}}\left(\ket{0}\otimes\ket{G}+\ket{0}\otimes\ket{G'}\right)$, where $G'$ is the graph obtained by applying $P$ to $G$. In this case, the $X$ measurement results in $-1$ with probability given by $\frac{1-|\braket{G'}{G}|^2}{2}$. As before, since $|\braket{G'}{G}|\leq \frac{1}{\sqrt{2}}$, this comparison operation yields an incorrect result with probability at most $\frac{3}{4}$. This gives an automorphism test with one-sided error of at most $\frac{3}{4}$.

\subsubsection{Vertex comparison}

Comparing two vertices is a special case of automorphism testing, where the permutation applied is a simple pairwise swap between the vertices, and hence can be achieved using the previous procedure. However, in this case, a procedure that gives test with a one sided error of at most $1/2$ exists.

Suppose, for now, that the two vertices, $a$ and $b$ do not share an edge, and that swapping $a$ and $b$ is an automorphism on $G$. Then, the two vertices share the exact same neighbours, other than self-loops. Also, if $a$ has a self-loop, then so does $b$.
Therefore the observable $X_aX_b$ is in the stabiliser of the graph state. This can be seen by simply multiplying $K_G^a$ by $K_G^b$. Thus, a measurement of this operator is guaranteed to result in $+1$. Now, suppose swapping $a$ and $b$ is \emph{not} an automorphism. Then either $N(a)\setminus a \neq N(b) \setminus b$, or one of $a$ and $b$ has a self-loop where the other does not. In the first case, there is an element of the stabiliser of the state that anti-commutes with the observable $X_aX_b$. In such a case, $N(a) \Delta N(b)$ (here $\Delta$ represents the symmetric difference between two sets $S_1 \Delta S_2 = (S_1 \cup S_2) \setminus (S_1 \cap S_2 )$) is not empty, and for any vertex $u$ in this set, the operator $K_G^u$ necessarily acts as $Z$ on either $a$ or $b$ and as identity on the other. Thus $K_G^u$ anti-commutes with $X_a X_b$, and hence the expectation value for the measurement is $0$. On the other hand, if $N(a) \setminus a = N(b) \setminus b$, but one of $a$ and $b$ has a self-loop where the other doesn't, then $-X_aX_b$ is a stabiliser of the graph state. In either case, measuring the observable $X_aX_b$, gives the result $-1$ with probability of at least $1/2$. Hence, individual vertices can be compared with one-sided error probability of $\frac{1}{2}$.

Now consider the case where vertices $a$ and $b$ do share an edge. In this case, $X_a X_b$ is no longer in the stabiliser under any circumstance, since it anti-commutes with $K_G^a$. However, if swapping $a$ and $b$ is an automorphism, then $Y_aY_b$ is in the stabiliser of $\ket{G}$, since this is simply the product of $K_G^a$ and $K_G^b$. Hence, measuring the observable $Y_aY_b$ will determine whether swapping $a$ and $b$ is automorphism while $(a,b) \in V$. Since both tests require only two single-qubit measurements, they can be accomplished with $O(1)$ operations.

\subsection{Read-out operations}
\subsubsection{Degree parity test}
In graph theory, an \emph{Euler graph} is one whose every vertex is of \emph{even} degree, meaning that each vertex has an even number of neighbours. Similar to the previous section we can devise a simple probabilistic algorithm that tests whether a graph is Eulerian, in $O(n)$ time, although this procedure requires two copies of the graph state.

For sake of clarity, let us consider first the simplified case where the graph tested is promised to not include any self-loops. Then, it is necessarily the case that
\begin{equation}\label{eq:euler-stab}
\overline{X} = \prod_{v \in V} X_v
\end{equation}
is a stabiliser of the graph state if and only if the parity of every vertex is even. Hence, performing this measurement on an Euler graph will always give the measurement outcome $+1$. On the other hand, if the graph is not Eulerian, there exists at least one vertex $v \in V$ such that $|N(v)| \equiv 1\mod 2$. Then, the stabiliser 
\begin{equation}\label{eq:euler-nocomm}
K_{G}^{v}=(-1)^{|N(v) \cap \{v\}|}X_{v}\prod_{u\in N(v)\setminus u} Z_u.
 \end{equation}
\emph{anti-commutes} with the observable $\overline{X}$. Hence, performing the measurement in Eq.\ \ref{eq:euler-stab} on a non-Euler graph state gives the measurement outcomes $+1$ and $-1$ with equal probability. Hence, by performing said measurement on the graph state, we can distinguish these two cases with one-sided error of $\frac{1}{2}$. This gives an $O(n)$ run-time algorithm.

Now, let us consider the general case in the (possible) presence of self-loops. A self-loop adds two to the degree of the vertex to which it is incident. This entails that a multigraph $G$ with self-loops is Eulerian if and only if the \emph{simple} graph $G'$ obtained from $G$ by removing all self-loops is Eulerian. Recall that a self-loop on a vertex $v$ is represented in our quantum data structure simply as a Pauli $Z$ operator acting on $v.$ Hence, an Euler graph $G$ has a stabiliser of the form
\begin{equation}\label{eq:euler-self-loops}
	\prod_{v \in V} (-1)^s X_v,
\end{equation}
where $s$ is the number of self-loops in $G.$ Hence, in the general case one can still proceed with the procedure outlined above of measuring the observable in Eq.\ \ref{eq:euler-stab} on two copies of $G$. If the graph is Eulerian, then the measurement results will be consistently either $+1$ or $-1$---depending on $s$, but always one or the other on all copies of $G$. And, for non-Euler graphs the measurement result will be randomly and independently $+1$ or $-1$, resulting in mismatched results with probability $\frac{1}{2}$. Hence, this procedure provides a test of whether or not the graph is Eulerian with one-sided error of $\frac{1}{2}$. Furthermore, the procedure also gives one extra bit of information beyond whether the graph is Eulerian or not, namely the parity of $s$.

A similar algorithm for testing whether the degree of every vertex is \emph{odd} can also be devised by substituting the observable in Eq.\ \ref{eq:euler-stab} with the operator $\overline{Y} = \prod_{v \in V} Y_v$.

\subsubsection{Recovering classical representation}

Finally, we turn to the issue of recovering the encoded graph for a particular graph state $\ket{G}$. As any classical data structure for an arbitrary graph of $|V|$ vertices and $|E|$ edges requires at least $\log_2 \binom{|V|^2}{|E|}$ bits to uniquely label the graph, and to encode all graphs of $|V|$ vertices requires $|V|^2$ bits. As a result, Nayak's theorem \cite{nayak99} implies that the probability of correctly identifying the graph encoded by $\ket{G}$ is exponentially small in $|V|$. Even if multiple copies of the graph state are given, the joint states are not in general perfectly orthogonal for different graphs, and hence the encoded graph cannot be deterministically decided \cite{chefles2000quantum}. However, we now show that it is possible to recover the encoded graph with high probability given a number of copies of $\ket{G}$ linear in $|V|$.

\begin{algorithm}[t]
\caption{Graph State Readout}\label{alg:readout}
\begin{description}
\item[Input:] An expected $K = O(|V|)$ copies of the graph state $\ket{G}$ representing the graph $G$.
\item[Output:] A classical description of $G$.
\item[Steps:]
\end{description}
\begin{enumerate}
\item Let $k = 1$ and $S = \{\}$ be an empty set. Let $\mathbb{A}$, $\mathbb{B}$ and $\mathbb{C}$ be three $|V|$-qubit registers.
\newline 
\newline
\textbf{while} $|S| \neq |V|$:
\begin{enumerate}
\item Take two copies of the state $\ket{G}$, and place them in registers $\mathbb{A}$ and $\mathbb{C}$. Prepare register $\mathbb{B}$ in state $\ket{+}^{\otimes |V|}$.

\item \textbf{for} $1\leq i \leq |V|$:
\begin{enumerate}
\item Apply a $\CZ$ between the $i$th qubit of $\mathbb{A}$ and the $i$th qubit of $\mathbb{B}$ and a second $\CZ$ between the $i$th qubit of $\mathbb{B}$ and the $i$th qubit of $\mathbb{C}$.
\item Measure qubit $i$ of $\mathbb{A}$ in the $X$ basis. Record the results in a bit vector $a_k^i$.
\item\label{step:graph-linked} Measure qubit $i$ of $\mathbb{B}$ in the $X$ basis. Record the results in a bit vector $b_k^i$.  
\item Measure qubit $i$ of $\mathbb{C}$ in the $X$ basis. Record the results in a bit vector $c_k^i$.
\end{enumerate}
\item \textbf{if} $S \cup b_k$, wehre $b_k = b_k^1 \ldots b_k^{|V|}$, is linearly independent add $b_k$ to $S$ and increment $k$.
\end{enumerate}
\item Let $A$, $B$ and $C$ be square matrices such that the $k$th column of each is $a_k$, $b_k$ or $c_k$ respectively. Solve the system of linear equations given by
$$b_k^T x_i = a_k^i + c_k^i ~\text{mod}~2~\forall~i,k$$
for each $x_i$ and denote by $\Lambda$ the $|V|\times |V|$ matrix for which the $i$th column is given by $x_i$. $\Lambda$ should be a symmetric matrix.
\item Take a final copy of $\ket{G}$ and apply $\CZ$ between every unordered pair of qubits $(i,j)$ such that $\Lambda_{ij} = \Lambda_{ji} =1$. Measure the resulting state in the $X$ basis and denote the results $d_k$.
\newline
The adjacency matrix for the graph $G$ is given by $\Gamma = \Lambda + D$, where $D$ is the diagonal matrix such that $D_{ii} = d_k$.\label{step:detangle}
\end{enumerate}

\end{algorithm}

The procedure is presented as Algorithm \ref{alg:readout}. The algorithm proceeds in three main steps. First, a maximal set of linear equations over $GF(2)$ is constructed via repeated measurements of copies of the state $\ket{G}$. This is achieved by first linking two copies of the graph state using $\CZ$ operators via an intermediary set of ancilla qubits. Then all qubits are all measured in the $X$ basis. In the second step, these measurement results are used to recover the adjacency matrix (less self-edges) for $G$ using linear algebra. The final step recovers the self edges via an additional adapted measurement on $\ket{G}$.

We now analyse the number of copies of the graph state required by the algorithm. Whether an execution of the main loop adds a new vector to $S$ depends on the measurement outcomes on the ancilla qubits, stored in the bit vector $b_k$, and on $S$. Since each measurement has the same probability of resulting in either $0$ or $1$, and the number of distinct non-zero vectors which can be generated from linear combinations of the vectors in $S$ is $2^{|S|}-1$, the probability that a random set of measurements vector linearly independent from those already in $S$ is $1 - 2^{|S|-|V|}$. Hence, the expected number of times the loop must iterate in order to add a single vector to $S$ is $(1 - 2^{|S|-|V|})^{-1}$. Adding over all values of $|S|$ from $0$ to $|V|-1$ gives a total expected runtime for the first loop of
\[
\sum_{\ell=0}^{|V|-1} \frac{1}{1-2^{\ell-|V|}} \leq 2|V|.
\]
Hence, the total expected number of copies of the state $\ket{G}$ needed is bounded from above by $4|V|+1$.

The correctness of the algorithm can be verified by analysing a single run through the first loop. After step \ref{step:graph-linked} the state of unmeasured qubits is
\begin{align*}
\ket{\Phi_G} &= 
\CZ^{G} \left(\bigotimes_{j=1}^{|V|} X^{b_k^j}\right) H^{\otimes |V|} \left(\bigotimes_{i=1}^{|V|} X^{a_k^i}\right) H^{\otimes |V|} \CZ^{G}\ket{+}^{\otimes |V|}\\ 
&=  \overline{Z}^{N(b_k)} \left(\bigotimes_{i=1}^{|V|} Z^{a_k^i}\right)\ket{+}^{\otimes |V|},\\
\end{align*}
where $\CZ^G$ represents the product of $\CZ$ operations wherever there is an edge in $G$ (i.e. $\ket{G} = \CZ^{G} \ket{+}^{\otimes |V|}$), and $\overline{Z}^{N(b_k)} = \prod_i {Z}^{N(b_k^i)}$ where ${Z}^{N(b_k^i)}$ represents a $Z$ operator on all neighbours in $G'$ of vertex $i$ if and only if $b_k^i = 1$. The graph $G'$ is taken to be the result of removing all self-edges from $G$.

Hence, measuring $\ket{\Phi_G}$ and recording the results as $c_k$ gives results in
\begin{align*}
c_k^i =& b_k^T x_i + a_k^i~\text{mod }2
\end{align*}
where the vector $x_i$ is defined such that $x_i^j = 1$ if and only if $j$ is a neighbour of $i$ in $G'$. Thus, $\Lambda$ is the adjacency matrix of $G'$. Since $b_k$ are linearly independent, the system of equations has a unique solution, and hence $\Lambda$ can always be found.

Step \ref{step:detangle} determines the location of the self-edges in $G$ by removing all edges in common between $G$ and $G'$. The resulting state is a tensor product of $X$ eigenstates, such that $d_k = 1$ if and only if vertex $k$ has a self-edge. Thus $D$ is the adjacency matrix for the graph obtained from $G$ by removing all edges which join vertices to other vertices, leaving only self-edges. Combining $D$ and $\Lambda$ by addition, then, results in the adjacency matrix for $G$. 

\section{Comparison to classical data structures}

Introducing a quantum data structure for classical information raises the obvious question of whether it presents any advantages over using traditional classical data structures---in the case of graphs, a quantum data structure would be expected to offer advantages over adjacency matrices, incidence lists and other classical graph representations. It is certainly the case that the use of graph states is competive with the most widely studied classical graph data structures in terms of the operations studied here. In particular graph states offer more efficient comparison and complementation operations than either incidence matices, incidence lists, adjacency lists or adjacency matrices, while at the same time remaining more space efficient. However, the set of operations implementable with a single copy of a graph state are more limited than those allowed by such classical structures, and it is natural to ask whether some specialized classical structure could offer similar or better performance for the set of operations admitted by graph states. We now answer this question in the negative.

\begin{theorem}
No classical graph data structure requiring $O(N^{2-\epsilon})$ space can allow for comparison between arbitrary graphs of $N$ vertices with error bounded below $\frac{1}{2}$ for any constant $\epsilon>0$.
\end{theorem}
 \begin{proof}
The proof of this statement follows from a simple counting argument. There are $N^2$ possible edges within a graph, and each of these can be either present or absent, leading to a total of $2^{N^2}$ graphs of $N$ vertices. If the graph is represented using a string of $n$ bits, for $n<N^2$, then there must exist a single string which corresponds to $\lceil 2^{N^2-n}\rceil$ distinct graphs. These graphs are hence indistinguishable. Thus, only representations requiring at least $N^2$ bits can allow for unambiguous discrimination of all graphs even with bounded error.
 \end{proof}
 
\begin{theorem}
No classical graph data structure can have both an edge complementation operation and a vertex comparison operation that require only $O(1)$ elementary gates.
\end{theorem}
 \begin{proof}
 We proceed by contradiction. Consider a graph of $N$ vertices. If complementing an edge requires $O(1)$ gates, this means that every complementation operator, regardless of input, acts on at most $k$ bits for some constant $k$. Likewise, if vertex comparison requires $O(1)$ gates, then every comparison operator, regardless of input, acts on at most $c$ bits for some constant $c$. We label the $k$-bit set that the complementation operator acts upon, when input the vertices $A$ and $B$ as $S_e(A,B)$. Similarly, the $c$-bit set that the comparison operator acts upon when input the vertices $A$ and $B$ we will call $S_c(A,B)$.
 
Now, when an edge incident to vertex $A$ is complemented, it can change whether swapping it with $B$ is a graph automorphism (except for the case where the edge being complemented is also incident to $B$). Therefore, the set $S_c(A,B)$ must include at least one of the $k$ bits in set $S_e(A,X)$ for every vertex $X$. 
Hence, the average number of times a bit in $S_c(A,X)$ is in $S_e(A,Y)$, taken over all values of $X$ and $Y$, is at least $(N-1)/c$. Although the sum of the cardinality of the $S_e$ sets is $k{N \choose 2}$, the previous constraint implies that the cardinality of the union of $S_e(A,B)$ over all $A$ and $B$ is linear in $N$, and hence the number of bits examined by the comparison operations is also linear in $N$. 

We now show that this is too few bits to determine whether swapping a pair of vertices is a graph automorphism, by counting the number of ways in which vertices can be partitioned such that swapping vertices within a partition constitutes an automorphism on the graph. If there are $m$ distinct types of vertices, then the number of ways to partition $N$ vertices into $m$ sets is giving by
 
 \begin{equation}
 \frac{N!}{n_1!n_2!....n_m!m!}.
 \end{equation}
Thus $\log_2 \frac{N!}{n_1!n_2!....n_m!m!}$ bits are required to determine the partitioning uniquely, which is required in order to correctly determine which swaps are automorphisms.  Consider the case when $n_i=2$ for $i=1,2,...m$, then we have
  \begin{align}
  \log_2 \frac{N!}{n_1!n_2!....n_m!m!} 
  &= \log_2 \frac{N!}{2^{N/2} (\frac{N}{2})!}.
  \end{align}
Expanding this equation using Stirling's approximation, we find that the dominant term scales is $\frac{N}{2}\log_2 \frac{N}{2}$, and hence the number of bits required to describe such a partitioning scales superlinearly in $N$. Hence we have a contradiction with our earlier requirement that this quantity scale linearly in $N$.
\end{proof}

Similar dichotomy results can be proven for other operations on classical data structures which are permitted on graph states.

\section{Conclusions}

Graph states, with their close connection to graphs, represent a class of quantum states for which very naturally embody a mathematical structure common to many computational tasks. Here we have shown that this representation allows for efficient manipulation of this structure, and hence such states can be used as a data-structure which outperforms any possible classical counterpart. While the range of operations are more limited than in the case of classical data-structures, due to linearity constraints, we believe that this approach offers a new primitive for quantum algorithms.

Most quantum data structures have a classical counterpart upon which they are based, \emph{e.g.} the qubit is the quantum analogue of the classical bit. In contrast, there is no classical analogue of quantum graph state. We believe this motivates the further study of (intrinsically) quantum data structures. It is possible that further study into these structures may uncover further problems for which---like factoring---quantum computers offer a definite advantage.

\section{Acknowledgements}

The authors thank Simon Perdrix for useful discussions, and Yingkai Ouyang and Joshua Kettlewell for helpful comments on the manuscript. This material is based on research funded by the Singapore National Research Foundation under NRF Award NRF-NRFF2013-01.
 
 \bibliographystyle{apsrev}
 \bibliography{mybib}

\end{document}